\documentclass[dvipsnames,letterpaper,11pt]{article}
\usepackage[margin=1in]{geometry}
\usepackage[utf8]{inputenc}
\usepackage{graphicx,fullpage,paralist}
\usepackage{amsmath,amssymb,tabu,amsthm}
\usepackage{dsfont}
\usepackage{bm}
\usepackage{array}
\newcolumntype{C}{>{$}c<{$}}
\usepackage{comment,hyperref}
\usepackage{booktabs}
\usepackage{subcaption}
\usepackage[font={footnotesize}]{caption}
\usepackage{mathtools}
\usepackage{thmtools}
\usepackage{tikz}
\usepackage{tikz-network}
\usepackage{pgfplots}
\usepackage{varwidth}
\usepackage{xstring}
\pgfplotsset{compat=1.10}
\usepgfplotslibrary{fillbetween}
\usetikzlibrary{backgrounds}
\usetikzlibrary{patterns}
\usetikzlibrary{arrows}
\usetikzlibrary{calc}
\usetikzlibrary{matrix,decorations.pathreplacing,positioning,fit}
\usepackage{enumitem}
\usepackage{blkarray}
\usepackage{csquotes}
\usepackage{stmaryrd}
\usepackage[skip=3pt, indent=20pt]{parskip}

\usepackage{anyfontsize}
\newcommand{\PP}{\mathbb{P}}

\newcommand{\calD}{\mathcal{D}}

\newcommand{\RR}{\mathbb{R}}

\newcommand{\EE}{\mathbb{E}}

\DeclareMathOperator{\OPT}{\normalfont\text{OPT}}
\DeclareMathOperator{\opt}{\normalfont\text{OPT}}

\DeclareMathOperator{\ALG}{\normalfont\text{ALG}}

\newtheorem{theorem}{Theorem}
\newtheorem{lemma}{Lemma}

\usepackage[textsize=tiny,textwidth=2cm]{todonotes}

\usepackage{multirow}
\usepackage[noend]{algpseudocode}
\usepackage{algorithm,algorithmicx}

\newlength{\algofontsize}
\hypersetup{
	colorlinks,
	linkcolor={red!50!black},
	citecolor={blue!50!black},
	urlcolor={blue!80!black}
}
\begin{document}
\algrenewcommand\algorithmicrequire{\textbf{Input:}}
\algrenewcommand\algorithmicensure{\textbf{Output:}}
	
\title{Prophet Upper Bounds for Online Matching and Auctions
\vspace{.5cm}
}
	
\author{
Jos\'e Soto			
\thanks{Department of Mathematical Engineering, Universidad de Chile.}
\thanks{Center for Mathematical Modeling IRL-CNRS 2807, Universidad de Chile.}
\and Victor Verdugo
\thanks{Institute for Mathematical and Computational Engineering, Pontificia Universidad Católica de Chile.}
\thanks{Department of Industrial and Systems Engineering, Pontificia Universidad Católica de Chile.}
}

\date{}

\maketitle
\thispagestyle{empty}
\begin{abstract}
In the online 2-bounded auction problem, we have a collection of items represented as nodes in a graph and bundles of size two represented by edges. Agents are presented sequentially, each with a random weight function over the bundles. The goal of the decision-maker is to find an allocation of bundles to agents of maximum weight so that every item is assigned at most once, i.e., the solution is a matching in the graph. When the agents are single-minded (i.e., put all the weight in a single bundle), we recover the maximum weight prophet matching problem under edge-arrivals (a.k.a. prophet matching). 

In this work, we provide new and improved upper bounds on the competitiveness achievable by an algorithm for the general online 2-bounded auction and the (single-minded) prophet matching problems. For adversarial arrival order of the agents, we show that no algorithm for the online 2-bounded auction problem achieves a competitiveness larger than $4/11$, while no algorithm for prophet matching achieves a competitiveness larger than $\approx 0.4189$. Using a continuous-time analysis, we also improve the known bounds for online 2-bounded auctions for random order arrivals to $\approx 0.5968$ in the general case, a bound of $\approx 0.6867$ in the IID model, and $\approx 0.6714$ in prophet-secretary model.
\end{abstract}
\thispagestyle{empty}
\newpage

\section{Introduction}

In many collective platforms and markets, decisions must be made on the fly, i.e., online. Marketplaces that customize prices dynamically and ridesharing platforms that match vehicles to users are two examples of online decision-making platforms. The policies that the platform can adopt should be fast and efficient and typically designed with limited information. Under this scenario, matchings and auctions are not just market design cornerstones but the very essence of the operation of many modern online platforms. Solutions proposed for school choice platforms, organ donation, ridesharing, recommender systems (such as those employed by Meta and Netflix), and job markets are all fundamentally based on matching policies. Similarly, dynamic pricing and allocation in collective market platforms such as Amazon, Mercado Libre, and eBay can be effectively modeled via combinatorial auctions, further underscoring the central role of these problems.

The prophet inequality, nowadays a central paradigm in the stochastic online algorithms literature, has been extensively studied during the last decade in the algorithms community~\cite{GM66,Hill1982,Kertz1986,SC83}. In its most basic single-choice form, a sequence of random values independently chosen from known distributions is presented sequentially, and the decision-maker must choose when to accept a value, making an irrevocable accept-or-reject decision. This stochastic selection setting has been extended to more complex combinatorial environments, such as matchings and auctions, to capture more involved constraints on the solution constructed by the decision-maker~\cite{Ehsani2018,Ezra2020OnlineModels,kleinberg2019matroid,Rubinstein2016}. More specifically, in the {\it 2-bounded auction problem}, we have a collection of items represented as nodes in a graph and bundles of size two represented by edges. Agents are presented sequentially, each with a random weight function over the item bundles. These weight functions are chosen from independent distributions, one per agent. The goal of the decision-maker is to find an assignment of bundles to agents of maximum weight so that every item is assigned at most once, i.e., the solution is a matching in the graph. When the agents are single-minded (i.e., put all the weight in a single bundle), we recover the maximum weight matching problem. Generally, combinatorial auctions is a central research topic in the intersection of algorithms and economics and has been intensively studied from an approximation and mechanism design perspective~\cite{assadi2021improved,assadi2020improved,babaioff2014efficiency,baldwin2019understanding,correa2023constant,dobzinski2016breaking,dobzinski2005approximation,feldman2013combinatorial,marinkovic2024online,paes2020computing}.

In this work, we study the approximation limits for the 2-bounded auction problem in the prophet model. We consider several variants regarding the online arrival regime of the agents: the classic adversarial arrival model, the prophet-secretary model (i.e., random arrival order), and the prophet IID model (i.e., identical weight distributions). In all these models, an algorithm is $\alpha$-competitive
if it returns a solution for the 2-bounded auction problem with a total expected
weight of at least an $\alpha$-fraction of the expected optimal {\it offline} weight, i.e., the one where the decision-maker has full information about the realizations. We provide new upper bounds on the worst-case competitiveness for the 2-bounded auction problem and the (single-minded) maximum weight matching problem in the prophet model (for short, prophet matching), including both adversarial and random order arrival models.

\subsection{Our Results}

We first consider the prophet model with general distributions, i.e., with an adversarial arrival order of the agents.
This model's previously best-known upper bound is $3/7\approx 0.42857$ by Correa et al. \cite{Correa2022CFPW}, which holds even in the single-minded (matching) case. 
Using a very simple instance, we improve this upper bound in the 2-bounded auction setting to $4/11\approx 0.363636$.
Our next upper bound is for the single-minded case, i.e., the prophet matching problem. 
We provide a new upper bound that improves the $3/7\approx 0.42857$ from \cite{Correa2022CFPW} down to $0.418928$.

\begin{theorem}\label{thm:upper-bounds-matching-adversarial}
The following upper bounds hold for the 2-bounded auction problem:
\begin{enumerate}[itemsep=0pt,label=\normalfont{(\alph*)}]
    \item No algorithm for the 2-bounded auction problem in the prophet model achieves a competitiveness larger than $4/11\approx 0.363636$.\label{ub-matching-a}
    \item No algorithm for the prophet matching problem achieves a competitiveness larger than $0.418928$.\label{ub-matching-b}
\end{enumerate}
\end{theorem}

We prove Theorem \ref{thm:upper-bounds-matching-adversarial} in Section \ref{sec:2bca-adversarial}. In our next set of results, we address the $2$-bounded auction problem for random arrival orders.
For the single-choice prophet-secretary problem,
Correa et al.~\cite{CorreaSZ2021} established an upper bound of $\sqrt{3} - 1 \approx 0.732$, and recently,  Bubna and Chiplunkar \cite{bubna2023prophet} improved this bound to $0.7254$. This also represented the previous best upper bound for prophet-secretary matching. We further reduce this upper bound for the prophet-secretary matching problem to $0.671355$. Our result also improves upon the previous best upper bound for the more challenging {\it single-sample} data-driven variant of prophet-secretary matching, which was the $\ln(2) \approx 0.6931$ construction by Correa et al. \cite{Correa2019ProphetDistribution} for the single-minded prophet IID problem with unknown distribution, which can be transformed into an upper bound for the single-sample prophet secretary model.

In the case of general 2-bounded auctions in the prophet secretary models and single-sample prophet secretary, the previous best upper bounds were $0.703$ by Huang et al.~\cite{HuangSY2022} and the $\ln(2) \approx 0.6931$ bound~\cite{Correa2019ProphetDistribution}. 
We improve both bounds to $0.596774$.
Finally, for general 2-bounded auctions in the IID model, the best-known upper bound was still the $0.7451$ bound from Hill and Kertz \cite{Hill1982} and Kertz \cite{Kertz1986}, which applies to single-minded agents and single-choice. We enhance this bound for 2-bounded auctions to $0.686641$. These results are summarized in the following theorem.

\begin{theorem}\label{thm:upper-bounds-matching}
The following upper bounds hold for the 2-bounded auction problem:
\begin{enumerate}[itemsep=0pt,label=\normalfont{(\alph*)}]
    \item No algorithm for the prophet-secretary matching problem achieves a competitiveness larger than $0.671355$.\label{ro-ub-matching-a}
    \item No algorithm for the 2-bounded auction problem in the prophet-secretary model achieves a competitiveness larger than $0.596774$.\label{ro-ub-matching-b}
    \item No algorithm for the 2-bounded auction problem in the prophet \normalfont\text{IID} model achieves a competitiveness larger than $0.686641$.\label{ro-ub-matching-c}
\end{enumerate}
\end{theorem}

We prove our new bounds from Theorem \ref{thm:upper-bounds-matching} in Section \ref{sec:2bca-ro}, using a continuous-time analysis to generate hard instances for the prophet-secretary and prophet IID settings. We present the main points of this approach in Section \ref{sec:rolling} by revisiting the $\sqrt{3}-1$ upper bound in \cite{CorreaSZ2021} for the single-choice prophet-secretary problem. 
Table~\ref{tab:table_1} summarizes our improved guarantees and the existing lower bounds (i.e., approximation guarantees).
\begin{table}[h!]
        \centering
        {\begin{tabular}{|c|c|c|c|}
        \hline
             2-bounded auctions  & Arrival model & Lower bounds & Upper bounds\\ 
            \hline
            \multirow{3}{*}{General case} & Prophet & $1/3\approx 0.333$ \cite{Correa2022CFPW,ma2020approximation} & $\approx 0.363636$ (this work)\\ 
                                          & Prophet-secretary& $1/3\approx 0.333$ \cite{Correa2022CFPW,ma2020approximation,marinkovic2024online} & $\approx 0.596774$ (this work)\\ 
                                          & IID & $(1-e^{-2})/2\approx 0.432$ \cite{marinkovic2024online} & $\approx 0.686641$ (this work) \\
            \hline
            \multirow{3}{*}{Matching} & Prophet & $\approx 0.337$ \cite{Ezra2020OnlineModels} & $\approx 0.418928$ (this work)\\ 
                                      & Prophet-secretary& $\approx 0.474$ \cite{MacRury2023} & $\approx 0.671355$ (this work) \\ 
                                      & IID & $\approx 0.745^*$ \cite{correa2021prophet} & $\approx 0.745$ \cite{Hill1982,Kertz1986} \\
            \hline
        \end{tabular}}
        \caption{{Known lower and upper bounds for 2-bounded auctions. In the IID matching model, all agents are single-minded, want the same edge,  and share the same distribution. Therefore this model coincides with the classic IID single-choice prophet inequality.}}
        \label{tab:table_1}
    \end{table}

\section{Preliminaries}\label{sec:prelim}
Given a graph $G=(V,E)$ and a set $A$ of $m$ agents, a function $M\colon A\to E\, \cup \{\bot\}$ is a {\it feasible allocation} if two different agents cannot be assigned to the same edge, and $\{M(a):a\in A\}\setminus \{\bot\}$ is a matching in $G$. When $M(a)=\bot$, agent $a$ has not been assigned any edge $e\in E$.
Every agent $a$ has a weight function $w_a\colon E\cup \{\bot\}\to \RR_+$ with $w_a(\bot)=0$, and the 2-bounded auction problem corresponds to
\begin{equation}
\max\Big\{\textstyle\sum_{a\in A}w_a(M(a))\colon M\text{ is a feasible allocation }\Big\},\label{eq:optimal-assignment}
\end{equation}
where every node in the graph represents an item to be sold, while every edge represents a bundle of two items. Therefore, \eqref{eq:optimal-assignment} corresponds to finding the feasible allocation of bundles to agents that maximize the total weight.

We denote by $\opt$ the optimal value of \eqref{eq:optimal-assignment}.
We remark that \eqref{eq:optimal-assignment} corresponds to finding a maximum weight matching on 3-uniform hypergraphs, which is, in general, NP-hard.
We say that the agents are single-minded if for every agent $a$, there exists a single edge $e_a\in E$ such that $w_a(e_a)>0$, and $w_a(e)=0$ for every $e\ne e_a$.
With single-minded agents, \eqref{eq:optimal-assignment} corresponds to finding the maximum weight matching, which can be solved efficiently.
This problem has been extensively studied in the literature, including larger bundle sizes (see, e.g., \cite{KesselheimRTV13,dutting2020prophet,Correa2022CFPW}).

Every agent $a$ has a distribution $\calD_a$ over the weight functions.
In the \textbf{prophet model}, the agents arrive in an arbitrary (possibly adversarial) order $a_1,\ldots,a_m$, and each $a_t$, independently from the rest, upon arrival reveals a random weight function $r_{a_t}\sim \calD_{a_t}$. 
The decision maker decides irrevocably whether to assign agent $a_t$ to some edge $e\in E$, otherwise it is assigned to $\bot$.
The assignment of agents should be a feasible allocation.
We are in \textbf{prophet-secretary model} when the agents arrive in random order.
We are in the \textbf{prophet IID model} when the distributions are all identical.
An algorithm is $\gamma$-competitive if, for every instance, it constructs a feasible assignment with a total weight that is, on expectation, at least a $\gamma$ fraction of the expected optimal value for the problem \eqref{eq:optimal-assignment}.

\section{Upper Bounds for Adversarial Arrivals: Proof of Theorem \ref{thm:upper-bounds-matching-adversarial}}\label{sec:2bca-adversarial}

We start by showing our new upper bound for 2-bounded auctions in the prophet model. Consider the graph obtained by adding a diagonal edge $\{a,c\}$ to the 4-cycle $\{\{a,b\}$, $\{b,c\}$, $\{c,d\}$, $\{a,d\}\}$.
Let $\varepsilon>0$ and consider three agents with weights distributions as follows:

\begin{enumerate}[label=(\roman*)]
    \item Agent 1 puts a weight of $3$ on each edge $\{a,b\}, \{b,c\}, \{c,d\}, \{a,d\}$, and zero on the edge $\{a,c\}$.
    \item Agent 2 selects uniformly at random exactly one edge $x$ from $\{a,b\}, \{b,c\}, \{c,d\}, \{a,d\}$ and puts a weight of $4$ on that edge; and zero on the rest.
    \item Agent 3, with probability $\varepsilon>0$, puts a weight of $4/\varepsilon$ on the diagonal $\{a,c\}$ and zero everywhere else; and with probability $1-\varepsilon$ reveals the zero weight function.
\end{enumerate}
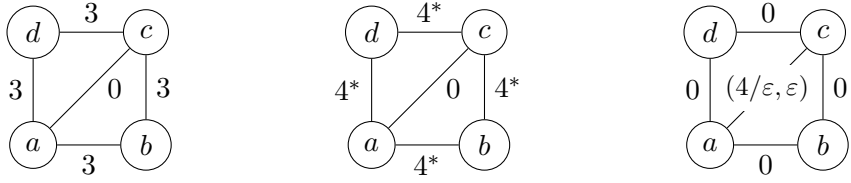
\begin{figure}[ht!]
\centering
\begin{tikzpicture}[scale=1.5]
  \node[circle, draw]  (a1) at (0,0) {$a$};
  \node[circle, draw]  (b1) at (1,0) {$b$};
  \node[circle, draw]  (c1) at (1,1) {$c$};
  \node[circle, draw]  (d1) at (0,1) {$d$};

  \draw (a1) -- node[below] {3} (b1) -- node[right] {3} (c1) -- node[above] {3} (d1) -- node[left] {3} (a1); 
  \draw (a1) -- node[right, xshift=0.1cm] {0} (c1);

  \begin{scope}[xshift=3cm]
    \node[circle, draw]  (a2) at (0,0) {$a$};
    \node[circle, draw]  (b2) at (1,0) {$b$};
    \node[circle, draw]  (c2) at (1,1) {$c$};
    \node[circle, draw]  (d2) at (0,1) {$d$};

    \draw (a2) -- node[below] {$4^*$} (b2) -- node[right] {$4^*$} (c2) -- node[above] {$4^*$} (d2) -- node[left] {$4^*$} (a2); 
  \draw (a2) -- node[ right, xshift=0.1cm] {0} (c2); 
 \end{scope}

  \begin{scope}[xshift=6cm]
    \node[circle, draw]  (a3) at (0,0) {$a$};
    \node[circle, draw]  (b3) at (1,0) {$b$};
    \node[circle, draw]  (c3) at (1,1) {$c$};
    \node[circle, draw]  (d3) at (0,1) {$d$};

   \draw (a3) -- node[below] {0} (b3) -- node[right] {0} (c3) -- node[above] {0} (d3) -- node[left] {0} (a3); 
  \draw (a3) -- node[midway,fill=white]{\small $(4/\varepsilon, \varepsilon)$} (c3); 
  \end{scope}
\end{tikzpicture}
 \caption{Agent 2 selects exactly one edge marked with a star at random and shows it with a weight equal to 4. Agent 3 puts a weight of $4/\varepsilon$ on the diagonal edge with probability $\varepsilon$.}
\end{figure}

\begin{proof}[Proof of Theorem \ref{thm:upper-bounds-matching-adversarial}\ref{ub-matching-a}]
If $\varepsilon$ is small enough, then the expected optimal offline solution is computed as follows. If agent 3 places a weight of $4/\varepsilon$ on the diagonal, then the optimum matching consists only of that edge, and so $\opt=4/\varepsilon$. If not, the optimum matching consists of assigning to agent 2 the unique edge $x$ of weight $4$ that she presents and assigning to agent 1 the edge opposite to $x$ of weight $3$, so that $\opt=4 + 3=7$.  
In total, by tending $\varepsilon$ to zero, the optimal offline solution has an expected weight of
$\EE(\opt) = \varepsilon \cdot (4/\varepsilon) + (1-\varepsilon)\cdot 7 \to 11$ when $\varepsilon \to 0$.

On the other hand, no algorithm can achieve a weight greater than $4$ if the agents arrive in the order of their names. Indeed, the only choices for an algorithm that does not select zero-weight edges are the following:
\begin{enumerate}[label=(\roman*)]
    \item Assign nothing to the first and second agents and wait for agent 3. This has an expected weight $(4/\varepsilon) \cdot \varepsilon = 4$.
    \item Assign nothing to the first agent and assign something to the second one. This also has an expected weight of $4$.
    \item Assign one edge (any) to the first agent and try to complete the matching with the edge that the second agent presents (succeeding with probability 1/4). This strategy has an expected weight of $1/4 \cdot (3+4)  + 3/4 \cdot 3 = 4$.
\end{enumerate}
Therefore, the competitiveness of any algorithm is at most $4/\EE(\opt) = 4/11\approx 0.363636$.
\end{proof}

Now, we move to the prophet-matching setting. Consider the graph $G$ with vertices $\{0,1,a,b,c,d\}$ and nine edges, arriving in the following order: 
$$\{a,0\}, \{a,1\}, \{b,0\}, \{b,1\}, \{c,0\}, \{c,1\}, \{d,0\}, \{d,1\}, \{0,1\}.$$
That is, $G$ is given by the full bipartite graph with sides $\{a,b,c,d\}$ and $\{0,1\}$, plus the extra edge $\{0,1\}$ depicted in Figure \ref{fig:pm-graph}. Let $p, q, r \in [0,1]$ be three fixed real values to be determined later, and such that  
\begin{align}
r=(1-p)q+(1-q)p, \quad p<q<r<1/2, \quad 1-p<2r. \label{eqn:bound1}
\end{align}
Let $\varepsilon>0$, and consider the following distributions for the weights of each edge:
\begin{align*}
w(\{a,0\}) \text{ and } w(\{a,1\})&: \text{deterministically weight $r$}. \\
w(\{b,0\}) \text{ and } w(\{b,1\})&: \text{deterministically weight $1-r$}. \\
w(\{c,0\}) \text{ and } w(\{c,1\})&: \text{weight $1-p$ with probability $q$, and 0 otherwise}. \\
w(\{d,0\}) \text{ and } w(\{d,1\})&: \text{weight $1$ with probability $p$, and 0 otherwise}. \\
w(\{0,1\}) &: \text{weight $1/\varepsilon$ with prob. $\varepsilon$, and 0 otherwise.}
\end{align*}

\begin{figure}[h!] 
\centering
\begin{tikzpicture}[scale=.9]\small
  \foreach \x/\y/\z in {a/4/0, b/2/2, c/2/4, d/4/6, 0/8/2, 1/8/4} {
    \node[circle, draw] (\x) at (\y, \z) {$\x$};
  }

      \draw (a) -- (0) node[pos=0.1, below right]
      {$(r,1)$};
      \draw (a) -- (1) node[pos=0.1, above left]
      {$(r,1)$};
      \draw (b) -- (0) node[pos=0.15, below]
      {$(1-r,1)$};
      \draw (b) -- (1) node[pos=0.15, above, xshift=-.5cm]
      {$(1-r,1)$};
 \draw (c) -- (0) node[pos=0.15, below, xshift=-.5cm]
      {$(1-p,q)$};
      \draw (c) -- (1) node[pos=0.15, above]
      {$(1-p,q)$};
 \draw (d) -- (0) node[pos=0.1, below left]
      {$(1,p)$};
      \draw (d) -- (1) node[pos=0.1, above right]
      {$(1,p)$};
  \draw (0) -- (1) node[pos=0.5, right]
      {$(1/\varepsilon,\varepsilon)$};
\end{tikzpicture}
 \caption{Each edge of the graph is labeled with the pair defining its weight and distribution.}
 \label{fig:pm-graph}
\end{figure}
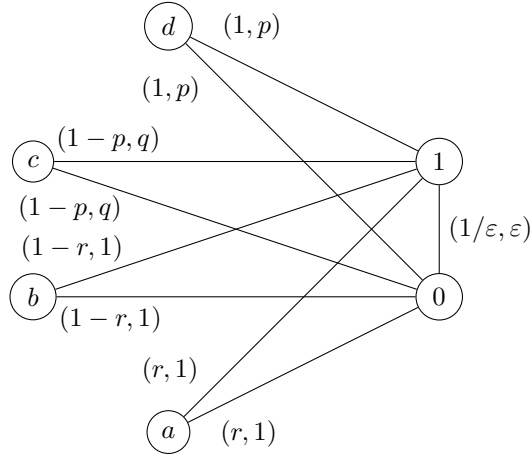

\begin{proof}[Proof of Theorem \ref{thm:upper-bounds-matching-adversarial}\ref{ub-matching-b}]
We say that an edge is \emph{active} if its weight is strictly positive. 
We claim that any online algorithm that outputs a matching can get a total weight of at most 1. 
Indeed, if the algorithm does not select anything incident to $a$, $b$, $c$, or $d$, then its maximum expectation is the expected weight of $\{0,1\}$, which is $1/\varepsilon \cdot \varepsilon=1$. We now study the complementary case, conditioning on the first edge $x$ the algorithm selects, assuming that the algorithm never selects an edge of weight zero.
\begin{enumerate}[label=(\roman*)]
    \item If $x$ is incident to $d$, the algorithm outputs $\{x\}$  obtaining an expected weight of $w(x)=1$.
    \item If $x$ is incident to $c$, then with probability $p$ an edge of weight 1 incident to $d$ can be used to increase the matching; therefore, the expected weight of the algorithm is $w(x)+p\cdot 1=(1-p)+p=1$.
    \item If $x$ is incident to $b$, then the expected achievable weight is $w(x)+q\cdot \max(1-p, 1\cdot p) + (1-q)\cdot 1\cdot p = (1-r) + q(1-p)+(1-q)p = (1-r) + r = 1$.
    \item If $x$ is incident to $a$, then the algorithm can create a matching with an edge incident to $b$, or wait until later. In the first case, it receives a weight of $w(x)+(1-r)=r + (1-r)=1$, and in the second case, the maximum expected weight it can obtain is $w(x)+q\cdot \max(1-p, 1\cdot p) + (1-q)\cdot 1\cdot p = r + q(1-p)+(1-q)p = 2r < 1$.
\end{enumerate}

Then, the expected weight of any algorithm is at most 1.
We now compute the expected weight of the optimal matching $\opt$ as follows.
If the edge $\{0,1\}$ is active, then the optimal matching consists only of the edge $\{0,1\}$, and so $\OPT=1/\varepsilon$. 
Therefore, for the rest of the analysis, we assume that $\{0,1\}$ is not active. 
Furthermore, we have the following facts for any realization: every matching using at most one edge has weight at most 1,
there is always a matching of weight 1 (taking the deterministic edges $\{a,0\}$ and $\{b,1\}$), and every matching with two nonzero edges uses different weights. 
Then, there are only six possible weights for matching of maximum weight. 
Using the inequalities of \eqref{eqn:bound1}, these values can be sorted in decreasing order as $1+(1-p) > 1+(1-r) > 1+r > (1-p)+(1-r) > (1-p)+r > 1$. 
Below, we study the probability that the optimal matching has a weight equal to each of these six cases. 
Note that this depends only on the status of the edges incident to $c$ and $d$.

\begin{enumerate}[label=(\roman*)]
    \item $\opt = 1+(1-p)=2-p$ in two cases: both edges incident to $d$ and at least one incident to $c$ are active; or exactly one edge incident to $d$ is active, and the only edge incident to node $c$ that can complete the matching is also active. 
    In total, this event occurs with probability $p^2(1-(1-q)^2)+ 2p(1-p)q$. 
    \item $\opt = 1+(1-r)=2-r$ in two cases: both edges incident to $d$ are active and no edge incident to $c$ is active; or if exactly only one edge incident to $d$ is active and the edge incident to $c$ that can complete the matching is not active. In total, this event occurs with probability $p^2(1-q)^2+2p(1-p)(1-q)$.
    \item \label{case3} $\opt$ is never equal to $1+r$, since whenever an edge of weight 1 is active, we can always complete the matching with an edge of weight $1-r$ incident to $b$, obtaining a matching of weight $1+(1-r)>1+r$.
    \item $\opt = (1-p)+(1-r)=2-p-r$ if both edges incident to $d$ are not active, and at least one edge incident to $c$ is active. This event occurs with probability $(1-p)^2(1-(1-q)^2)$
    \item $\opt$ is never equal to $(1-p)+r$ for the same reason as in case \ref{case3}.
    \item $\opt=1$, which occurs when all the edges incident to $c$ and $d$ are not active. This happens with probability $(1-p)^2(1-q)^2$
\end{enumerate}

Using $r=(1-p)q+(1-q)p$, and taking the limit when $\varepsilon\to 0$, we see that the optimum matching has an expected weight that converges to
\begin{equation*}
\begin{split}
&\lim_{\varepsilon\to 0} \varepsilon \cdot (1/\varepsilon) + (1-\varepsilon)\Bigl((2-p) (p^2(1-(1-q)^2)+ 2p(1-p)q)\\
&+  (2-r) (p^2(1-q)^2+2p(1-p)(1-q))+ (2-p-r) (1-p)^2(1-(1-q)^2) + 1 (1-p)^2(1-q)^2 \Bigr)\\
&= 2+2 p-3 p^2+p^3+2 q-10 p q+15 p^2 q-6 p^3 q-3 q^2+14 p q^2-19 p^2 q^2\\
&+6 p^3 q^2+q^3-4 p q^3+4 p^2 q^3.
\end{split}
\end{equation*}
Denote by $F(p,q)$ the polynomial obtained. 
The maximum of $F$ subject to the conditions in \eqref{eqn:bound1} is reached at $(p^*,q^*)\approx (0.299130, 0.364352)$. 
Thus, the competitive ratio is at most $1/F(p^*,q^*)\approx 0.418928$.
\end{proof}

\section{The Rolling Particle Approach}\label{sec:rolling}

To motivate the ideas behind our new random-order upper bounds (Theorem \ref{thm:upper-bounds-matching}), we revisit the upper bound of $\sqrt{3}-1$ for the prophet-secretary single-choice problem from \cite{CorreaSZ2021} to present our continuous time analysis. Suppose that $\lambda \geq 1$ is a fixed constant, and we have a very large number $m$ of \emph{unlikely} agents, each of them setting the weight of (the same) item $x$ to be $\lambda m$ with probability $1/m^2$ and zero otherwise. 
We also have an extra \emph{important} agent $A$ that weights item $x$ deterministically and equal to one.

By conditioning on the existence of an unlikely agent presenting a non-zero weight, the offline optimal weight is equal to $(1-(1-1/m^2)^m)\cdot \lambda m + (1-1/m^2)^m \cdot 1\to \lambda+1$ when $m\to \infty$.
If an online algorithm that sees the agents in random order ever sees an active unlikely agent, it must take it, as it would earn a profit of $\lambda m$. Then, the algorithm has only one decision: to skip or take the deterministic unit weight shown by the important agent $A$ when it appears (i.e., assigning the item to $A$). 
This decision depends on whether the expected weight that the algorithm can achieve on the remaining instance (i.e., from agents that have not arrived) is greater or smaller than one. 

Suppose every agent $a$ selects a position $T(a)$ uniformly at random in the interval $[0,1]$. 
In the limit, we can visualize the algorithm as a small rolling particle moving from left to right in the interval $[0,1]$, observing the agents at the positions they selected. 
For the instance we described, if we ignore the deterministic important agent $A$, in the limit, the algorithm recovers an expected weight of $\lambda$ from the unlikely agents as it moves from zero to one. The algorithm should take the unit deterministic weight from $A$ if and only if one is greater than the expected weight that the algorithm would get by skipping $A$, which is $\lambda$ times $1-T(A)$ (see Figure \ref{fig:prophet-secretary}). There exists a time $T^{\star}=1-1/\lambda$ such that if $T(A)$ is smaller than $T^{\star}$ then it is more convenient to skip $A$, and if $T(A)$ is larger, then it is more convenient to take the unit weight from $A$. 
If it skips $A$, the expected weight recovered is $\lambda$, and otherwise, the expected weight recovered is $\lambda T(A)+1$.
Therefore, the optimal algorithm recovers an expected weight of
\begin{align*}
\int_{0}^{T^{\star}} \lambda \mathrm dx + \int_{T^{\star}}^1 (\lambda x+1) \mathrm dx &= \lambda T^{\star}+ \lambda\cdot \frac{1-(T^{\star})^2}{2} + 1-T^{\star} =\lambda+\frac{1}{2\lambda}.
\end{align*}
Since the optimal offline expected weight is $\lambda+1$, the competitiveness of any algorithm on this instance is at most $(\lambda+1/(2\lambda))/(\lambda+1)$, which is minimized by setting $\lambda=(1+\sqrt{3})/2\approx 1.36603$, obtaining an upper bound of $\sqrt{3}-1\approx 0.732$.
We call this limit analysis a {\it rolling particle} approach, and we use it to present our new upper bounds for the 2-bounded auction problem.

\begin{figure}[H]
\centering
\begin{tikzpicture}[scale=2]
  \draw (0,0) -- (5,0);

\filldraw[black] (3,0) circle[radius=1pt];
  \filldraw[blue] (0,0) circle[radius=1pt];

  \filldraw[blue] (5,0) circle[radius=1pt];

  \node[below=0.3cm] at (0,0) {\small  0};

  \node[below=0.3cm] at (5,0) {\small  1};

  \node[below=0.3cm] at (3,0) {$x$};

  \draw[->, gray, dashed] (0,0.4) --  node[above]{$\lambda\cdot T(A)+1$} (3,0.4) -- (3,0.7);
  \draw[->, gray, dashed] (0,0.2) --  node[above, near end]{$\lambda$} (5,0.2) ;
\end{tikzpicture}
\caption{Rolling particle representation for the $\sqrt{3}-1$ upper bound in the prophet-secretary single-choice problem. When the algorithm sees agent $A$ at time $x=T(A)$, it has to decide whether to stop or continue. The expected earned weight of stopping is $\lambda x+1$, and the expected weight for not stopping is $\lambda$. If $x>T^{\star}$, the algorithm decides to stop.}
\label{fig:prophet-secretary}
\end{figure}
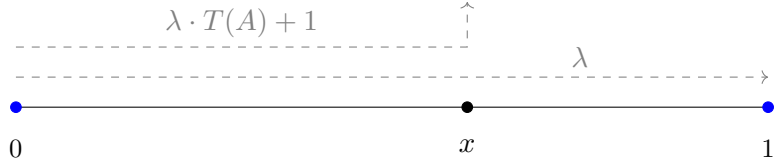

\section{Upper Bounds for Random Order and IID: Proof of Theorem \ref{thm:upper-bounds-matching}}\label{sec:2bca-ro}

In what follows, we provide our new upper bound for the prophet-secretary matching problem.
Let $\lambda\geq 1$ be a constant to be determined later and consider $m+2$ agents. Consider the graph with four vertices $\{a, b, c, d\}$, with one edge $\{a,b\}$, one edge $\{c,d\}$ and $m$ copies of edge $\{b,c\}$. 
Namely, this graph is a path with length three from node $a$ to $d$, but the middle edge $\{b,c\}$ is copied $m$ times.
Each edge is associated with a unique agent, as follows:
\begin{enumerate}
    \item Agent $L$ is interested in edge $\{a,b\}$ and deterministically puts a weight of one on it.
    \item Agent $R$ is interested in edge $\{c,d\}$ and deterministically puts a weight of one on it.
    \item Each of the other $m$ agents is unlikely: It sets the weight of its copy of $\{b,c\}$ to $\lambda m$ with probability $1/m^2$, and zero otherwise.  
\end{enumerate}
We refer to agents $L$ and $R$ as lateral agents. 
In this instance, the expected offline optimum, in the limit, gets all the weight from the unlikely agents associated with the middle edge copies and all the profit from the two lateral agents, that is, $\lambda+2$ overall (see Figure \ref{fig:prophet-secretary-2BCA}).

\begin{proof}[Proof of Theorem~\ref{thm:upper-bounds-matching}\ref{ro-ub-matching-a}]
The optimal algorithm behaves like a rolling particle that moves from zero to one, recovering a weight from the unlikely agents that is $\lambda$ times the distance traveled before stopping. 
As soon as it sees the first lateral agent, it has to decide whether to stop rolling and keep the weight from both lateral agents or to continue rolling. 
If it decides to continue, then when it sees the second lateral agent, it has to decide whether to stop and take it (getting a weight of one) or to continue rolling. 
From this, the optimal algorithm is defined by two thresholds $s,t$ with $0 \leq s \leq t \leq 1$. 
The particle stops on the second arriving lateral agent if and only if it has not stopped before, and the second arrives after time $t$. 
The particle stops on the first arriving lateral agent if and only if it arrives after time $s$.

Call $x\leq y$ to the random arrival times of the two lateral agents. The probability density of $x$ is given by minimum of two uniforms in $[0,1]$, and is equal by $2(1-x)$, and conditioned on $x$, the time $y$ is uniform in $[x,1]$. 
Therefore, given $s$ and $t$, the expected weight of the algorithm is
\begin{align*}
\ALG(s,t,\lambda)&=\int_{0}^{s}\left( \int_{x}^{t}  \frac{\lambda}{1-x}\mathrm dy + \int_{t}^{1}\frac{\lambda y+1}{1-x}\mathrm dy \right) 2(1-x)\mathrm dx + \int_{s}^1 (\lambda x+2)\cdot 2(1-x)\mathrm dx\\
&= \lambda s(2t-s) + s(1-t)(2+\lambda+\lambda t) +\frac{1}{3}(1-s)^2(6+\lambda+2\lambda s).
\end{align*}
Then, the competitiveness of the algorithm with thresholds $s$ and $t$ is $G(s,t,\lambda)=\ALG(s,t,\lambda)/(\lambda+2)$. Given $\lambda$, the we maximize $G(s,t,K)$ in all $0\leq s\leq t\leq 1$ to find the optimal thresholds $\overline s$ and $\overline t$.
By first-order conditions, and as long as $t\not\in \{0,1\}$, we impose $\partial_t G(\overline s,\overline t,\lambda)=0$ and get $\overline t=1-1/\lambda$, which is in $[0,1]$ if $\lambda \geq 1$. 
Then, by imposing $\partial_sG(\overline s,\overline t,\lambda)=0$ and $\overline s\leq \overline t$, we get $\overline s=1-\gamma/\lambda$, with $\gamma = 1+1/\sqrt{2}$, as long as $\lambda \geq \gamma$. 
Then, the best competitiveness that an algorithm can obtain for a given $\lambda\geq \gamma$ is $G(\overline{s},\overline t,\lambda)=(1+\sqrt{2}+3\lambda+3\lambda^3)/(3\lambda^2(\lambda+2))$.
This expression is minimized by setting $\lambda\approx 2.27861$, attaining a value of $0.671355$.
\end{proof}
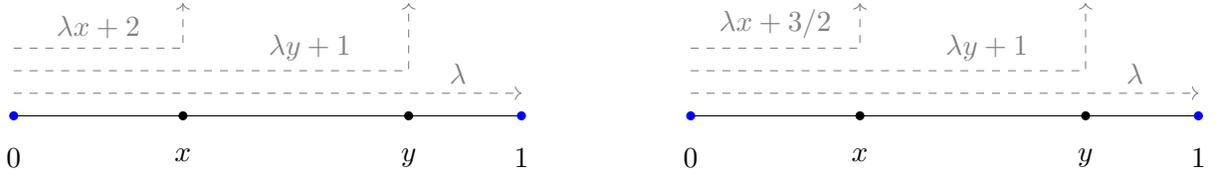
\begin{figure}[H]
\centering
\begin{tikzpicture}[scale=1.5]
  \draw (0,0) -- (4.5,0);

  \filldraw[black] (1.5,0) circle[radius=1pt] (0.08);
  \filldraw[black] (3.5,0) circle[radius=1pt] (0.08);
  \filldraw[blue] (0,0) circle[radius=1pt] (0.1);

  \filldraw[blue] (4.5,0) circle[radius=1pt] (0.1);

  \node[below=0.3cm] at (0,0) {0};

  \node[below=0.3cm] at (4.5,0) {1};

  \node[below=0.3cm] at (1.5,0) {$x$};
  \node[below=0.3cm] at (3.5,0) {$y$};

  \draw[->, gray, dashed] (0,0.6) --  node[above]{$\lambda x+2$} (1.5,0.6) -- (1.5,1);
  \draw[->, gray, dashed] (0,0.4) --  node[above, near end]{$\lambda y+1$} (3.5,0.4) -- (3.5,1) ;
  \draw[->, gray, dashed] (0,0.2) --  node[above, very near end]{$\lambda$} (4.5,0.2) ;

  \begin{scope}[xshift=6cm]
  \draw (0,0) -- (4.5,0);

  \filldraw[black] (1.5,0) circle[radius=1pt] (0.08);
  \filldraw[black] (3.5,0) circle[radius=1pt] (0.08);
  \filldraw[blue] (0,0) circle[radius=1pt] (0.1);

  \filldraw[blue] (4.5,0) circle[radius=1pt] (0.1);

  \node[below=0.3cm] at (0,0) {0};

  \node[below=0.3cm] at (4.5,0) {1};

  \node[below=0.3cm] at (1.5,0) {$x$};
  \node[below=0.3cm] at (3.5,0) {$y$};

  \draw[->, gray, dashed] (0,0.6) --  node[above]{$\lambda x+3/2$} (1.5,0.6) -- (1.5,1);
  \draw[->, gray, dashed] (0,0.4) --  node[above, near end]{$\lambda y+1$} (3.5,0.4) -- (3.5,1) ;
  \draw[->, gray, dashed] (0,0.2) --  node[above, very near end]{$\lambda$} (4.5,0.2) ;
  \end{scope}
  
\end{tikzpicture}
\caption{Instances in the prophet-secretary model for matching (left) and general 2-bounded auctions (right). 
The expected weight recovered by stopping on the first lateral agent at time $x$ is $\lambda x+2$ for matching and $\lambda x + 3/2$ for auctions. 
The expected weight recovered by stopping on the second lateral agent at time $y$ is $\lambda y+1$, and for not stopping is $\lambda$.}
\label{fig:prophet-secretary-2BCA}
\end{figure}

We now improve on the previous bad instance for prophet-secretary matching by working on the general setting of 2-bounded combinatorial. 
Let $\lambda\geq 1$ be a constant, and suppose that we have $m+2$ agents. 
Consider the graph with four vertices $\{a, b, c, d\}$ and the five edges $\{a,b\}$, $\{b,c\}$, $\{c,d\}$, $\{a,d\}$ and $\{b,d\}$.
Namely, the graph is a 4-cycle with an extra diagonal edge $\{b,d\}$.
The agents are described as follows:
\begin{enumerate}
    \item Agent $L$ shows with probability 1/2, edges $\{a,b\}$ and $\{b,c\}$ with weight 1 each and 0 on the rest. With probability 1/2, it shows edges $\{c,d\}$ and $\{a,d\}$ with weight 1 and 0 on the rest.
    \item Agent $R$ shows with probability 1/2, edges $\{b,c\}$ and $\{c,d\}$ with weight 1 each and 0 on the rest. With probability 1/2, it shows edges $\{a,d\}$ and $\{a,b\}$ with weight 1 and 0 on the rest.
    \item Each of the other $m$ agents is unlikely: With probability $1/m^2$, it sets the weight of the diagonal $\{b,d\}$ to $\lambda m$ and 0 to the rest. Otherwise, it shows the zero weight function.
\end{enumerate}

Agents $L$ and $R$ are called lateral agents. 
Observe that all edges of the graph are adjacent to the diagonal $\{b,d\}$. 
Therefore, if an algorithm assigns the diagonal edge at any time, it cannot increase the matching later. 
Furthermore, observe that no matter what pair of edges $\{e,f\}$ agent L shows (with non-zero weight), agent R will show an edge $z\not\in \{e,f\}$ such that $\{e,z\}$ or $\{f,z\}$ is a matching in the graph. 
Then, in the offline problem, we can always get a weight of 2 from agents L and R. 
In the optimal offline solution, we get a weight $\lambda m$ if at least one unlikely agent is active (i.e., shows a non-zero weight), and a profit of 2 otherwise (by selecting the appropriate edges). 
Hence the expected value of the offline optimal solution is
$(1-(1-1/m^2)^{m})\lambda m + (1-1/m^2)^m 2 \to \lambda+2$ when $m\to \infty$.

\begin{proof}[Proof of Theorem~\ref{thm:upper-bounds-matching}\ref{ro-ub-matching-b}]
An online algorithm, represented by a rolling particle that moves from zero to one, as soon as it sees the first lateral agent, it has to decide whether to stop rolling and select an edge from the first lateral agent or to continue. 
The expected weight for stopping at the first lateral agent is 3/2: it gets one unit from the current agent plus one times the probability that the other lateral agent shows a set of edges that contains one that is disjoint from the one already selected, which is 1/2. If it decides to continue, it keeps collecting weight from the uncertain agents until it sees the second lateral agent. 
Then it has to decide whether to stop (getting an extra weight of 1) or to continue until the end of the line (see Figure \ref{fig:prophet-secretary-2BCA}).

As for the prophet-secretary matching case, there are two thresholds $s,t$ with $0\leq s\leq t\le 1$ that the algorithm sets. 
The particle stops on the second lateral agent if and only if it has not stopped before, and the second one arrives after time $t$. 
The particle stops on the first arriving lateral agent if and only if it arrives after time $s$. 
Given $s$ and $t$, the expected weight of the algorithm solution is
\begin{align*}
\ALG(s,t,K)&=\int_{0}^{s}\left( \int_{x}^{t} \frac{\lambda}{1-x}\mathrm dy + \int_{t}^{1}\frac{\lambda y+1}{1-x}\mathrm dy \right) 2(1-x)\mathrm dx + \int_{s}^1 (\lambda x+3/2)\cdot 2(1-x)\mathrm dx\\
&= Ks(2t-s) + s(1-t)(2+\lambda+\lambda t) +\frac{1}{6}(1-s)^2(9+2\lambda+4\lambda s).
\end{align*}
The competitiveness of the algorithm with thresholds $s$ and $t$ is $G(s,t,\lambda)=\ALG(s,t,\lambda)/(\lambda+2)$. 
Given $\lambda$, the best algorithm is obtained by maximizing $G(s,t,\lambda)$ over $0\leq s\leq t\leq 1$ so we find the optimal thresholds $\overline{s}$ and $\overline{t}$. 
Imposing $\partial_2G(\overline{s},\overline{t},K)=0$, we deduce that $\overline{t}=1-1/K$, which is in $[0,1]$ if $\lambda\geq 1$. Then, imposing $\partial_1 G(\overline{s},\overline{t},K)=0$ and $\overline{s}\leq \overline{t}$, one deduces that $\overline{s}=1-1/\lambda=\overline{t}$. 
Therefore, for a given $\lambda$, the best competitiveness achievable is 
$G(1-1/\lambda,1-1/\lambda,\lambda)=(6\lambda^3+6\lambda-1)/(6\lambda^2(\lambda+2)).$
This expression is minimized by setting $\lambda\approx 1.36987$, attaining a value of $0.596774$.
\end{proof}

In what follows, we construct the upper bound in the prophet IID setting.
For each positive integer value $q$, consider the 2-bounded auction instance $I_q$ in the prophet IID model defined as follows: The graph is a 4-cycle, we have $q$ agents, and each one selects their weight function independently, according to the following distribution: two edges in the 4-cycle are selected randomly, and we put a weight of one on both, and a weight of zero on the rest.
In our analysis, we use the following lemma.
\begin{lemma}\label{lem:2bca-iid}
The following holds:
\begin{enumerate}[itemsep=0pt,label=\normalfont{(\alph*)}]
    \item Every fixed edge in the 4-cycle appears with weight one for any agent with probability 1/2.
    \item For every $q\geq 1$, the expected weight of the solution for any online algorithm in the instance $I_q$ is upper-bounded by $2-2^{-q+1}$.
    \item The expected size $F(q)$ of a maximum matching in the instance $I_q$ is given as follows: $F(1)=1, F(2)=11/6$ and for $q\geq 3$, we have $F(q)=2-4/6^q$.
\end{enumerate}\end{lemma}

\begin{proof}
Let $e$ be any edge of the 4-cycle $C$ and let $x, y, z$ be the other edges, among the $\binom{6}{2}=3$ subsets of 2 edges of $C$, 3 contains $e$, namely $\{e,x\}, \{e,y\}, \{e,z\}$ so the probability that $e$ is shown with weight one by an agent is $3/6=1/2$. This shows the first claim.

To see the second claim, consider any online algorithm that never assigns zero weight edges and condition on the step $i \leq q$ in which it assigns for the first time an edge to an agent (if this never happens, then the solution has weight 0). The maximum weight that the algorithm may recover is then 1 plus the probability that the edge $e$ opposite to the edge chosen at time $i$ appears in any of the $q-i$ next steps; by the previous claim, this is $1 + (1-(1/2)^{q-i})=2-2^{-q+i}$. This is maximized for $i=1$, with a guarantee of $2-2^{-q+1}$.

We analyze the third part by cases. If $q=1$, then the optimum matching has only 1 edge.
Now, let $q=2$. 
There are two non-isomorphic sets of two edges that the first agent may receive: either both of them have one node, say $x$, in common, or both form a perfect matching $M$. Focus on the first case, and suppose that the first agent shows consecutive edges $f$ and $g$. Then, as long as the second agent shows any edge $e$ outside $f$ and $g$ we can form a perfect matching of weight 2, assigning $e$ to the second agent and its opposite edge (which is in $\{f,g\}$) to the first. This event happens as long as the second agent does not show $\{f,g\}$, so 5 out of 6 times, the optimum matching has size 2, and in the other case, it has size 1.  Focus on the second case now, and suppose that the first edge shows opposite edges $f$ and $h$. As long as the second edge shows any edge $e$ in $\{f,h\}$, we can form a perfect matching of weight 2, assigning $e$ to the second edge and its opposite edge to the first. Once again, this event happens as long as the second agent does not show exactly the perfect matching disjoint from $\{f,h\}$, which happens with probability 5/6. So, regardless of the realization for the first agent, there is a probability of exactly 1/6 that the optimum assignment is 1 (and 5/6 that it is 2). So $F(2)=2 + 1/6 \cdot 1 = 11/6 =2-1/6$.

Now suppose that $q\geq 3$. We observe first that if at least one agent $a$ shows two opposite edges (say $\{f,h\}$), then there is an assignment of weight 2. Indeed, let $b$ and $c$ be the other two agents and suppose that there is no assignment of weight 2 using an edge from $a$. The only way for this to happen is that $b$ and $c$ show exactly the other two edges of the 4-cycle, say $\{f',h'\}$. But then there is an assignment of weight 2 using one edge from $b$ and one from $c$, and we are done. So assume now that every agent shows pairs of consecutive edges and let $a$ be one of the agents. We saw in the proof for $q=2$ that as long as another agent shows any pair of edges different than what $a$ showed, there is a way to have an assignment of weight 2. It follows that the only way that there is no assignment of weight 2 with $q\geq 3$ is that every agent shows the same pair of edges and, furthermore, that the common pair of edges are consecutive. There are only 4 configurations (one per each pair of consecutive edges) out of the $6^q$ possibilities for all the agents in which this happens. So the expected weight of the optimum is $F(q)=(4/6^q) 1 + (1-4/6^q) 2 = 2-4/6^q$ as claimed.
\end{proof}

The instances $I_q$ are not sufficient to get a sufficiently strong upper bound, since for $q=1$ and $q=2$, we get a competitiveness of $1$ and $9/11$ respectively, and for $q\geq 3$ we get $(2-2^{-q+1})/(2-4/6^q)$ which converges to one as $q\to \infty$. 
To make the instances harder for the online algorithms, we incorporate the possibility that agents show a diagonal of the 4-cycle with very high weight and low probability. 

More specifically, consider the following instance.
We have a graph with five edges, obtained by a 4-cycle and one diagonal. 
Let $\lambda,\theta \geq 1$ be two constants.
We have $m$ agents, and each selects the weight function as follows.
The agent flips two coins independently: the first one is head with probability $1/m^2$, and the second one is head with probability $\theta/m$. 
If the first coin is head, the agent puts a weight of $\lambda m$ on the diagonal edge; 
if the first coin is a tail, the agent puts a weight of zero on the diagonal. 
If the second coin is head, the agent selects two edges from the 4-cycle uniformly at random and puts a weight of one on both and a weight of zero on the other two edges of the 4-cycle; 
if the second coin is a tail, the agent puts a weight of zero on every edge of the 4-cycle.
We denote this instance as $J_{m,\lambda,\theta}$.
Observe that if an agent obtains a tail in the first coin and a head in the second coin, it behaves as an agent in the instances $I_q$.

\begin{proof}[Proof of Theorem \ref{thm:upper-bounds-matching}\ref{ro-ub-matching-c}]
In an instance $J_{m,\lambda,\theta}$, an agent can show 0, 1, 2, or 3 edges with non-zero weight. 
As long as one agent shows a diagonal edge with weight $\lambda m$, which is at least 2 for large $m$, the optimal solution has weight $\lambda m$. 
Otherwise, the optimal solution has an expected weight of $F(q_m)$, where $q_m$ is the number of agents whose second coin is a head. 
Since the second coin is a head with probability $\theta/m$, the random variable $q_m$ follows a Binomial distribution with parameters $(m,\theta/m)$, and therefore,
\begin{align*}
\lim_{m\to \infty}\EE_{q_m}[F(q_m)]&= \lim_{m\to \infty}\PP[q_m=1]+\frac{11}{6}\lim_{m\to \infty}\PP[q_m=2]+\lim_{m\to \infty}\sum_{t=3}^m \left(2-\frac{4}{6^t}\right)\PP[q_m=t]\\
&= e^{-\theta}\theta+\frac{11}{6}\cdot \frac{e^{-\theta}\theta^2}{2}+ \sum_{t=3}^{\infty} \left(2-\frac{4}{6^t}\right)\frac{e^{-\theta}\theta^t}{t!}\\
&=2+\frac{e^{-\theta}}{36}\Big(72-144e^{\theta/6}-12\theta-\theta^2\Big),
\end{align*}
where the first equality holds by Lemma \ref{lem:2bca-iid}, the second equality holds since the Poisson limit theorem implies that $\PP[q_m=t]$ converges to $e^{-\theta}\theta^t/t!$ for every $t\ge 1$ when $m\to \infty$, and the third holds by evaluating the summation and taking limit. Therefore, by denoting as $\opt_m(\lambda,\theta)$ the size of the largest matching in the instance, we have
\begin{align}
\lim_{m\to \infty}\EE[\opt_m(\lambda,\theta)]&=\lim_{m\to \infty}\left(1-\left(1-\frac{1}{m^2}\right)^m\right)\lambda m + \lim_{m\to \infty}\left(1-\frac{1}{m^2}\right)^m\EE_{q_m} [F(q_m)]\notag\\
&= \lambda +2+\frac{e^{-\theta}}{36}\Big(72-144e^{\theta/6}-12\theta-\theta^2\Big).\label{iid-lim-opt}
\end{align}
In what follows, we use the rolling particle approach to analyze the performance of an online algorithm in this family of instances. This time, the algorithm will encounter a random number $q$ of \emph{important agents} distributed according to a Poisson with parameter $\theta$  (i.e., the agents whose second coin is head), and in which it has to decide whether to stop or continue. 

\begin{figure}[H]
\centering
\begin{tikzpicture}[scale=1.5]
  \draw (0,0) -- (8,0);

  \filldraw[black] (0.3,0) circle[radius=1pt] (0.08);
  \filldraw[black] (4.5,0) circle[radius=1pt] (0.08);
  \filldraw[black] (6,0) circle[radius=1pt] (0.08);

  \filldraw[blue] (0,0) circle[radius=1pt] (0.1);

  \filldraw[blue] (8,0) circle[radius=1pt] (0.1);

  \node[below=0.3cm] at (0,0) {0};

  \node[below=0.3cm] at (8,0) {1};

  \node[below=0.3cm] at (4.5,0) {$t$};


  \draw[->, gray, dashed] (0,0.8) --  node[above]{} (0.3,0.8) -- (0.3,1.4);
  \draw[->, gray, dashed] (0,0.6) --  node[above]{$\lambda t + 2-\exp(-\theta t/2)$} (4.5,0.6) -- (4.5,1.4);
  \draw[->, gray, dashed] (0,0.4) --  node[above, near end]{} (6,0.4) -- (6,1.4) ;
  \draw[->, gray, dashed] (0,0.2) --  node[above, very near end]{$\lambda$} (8,0.2) ;

\end{tikzpicture}
\caption{The number of important agents is distributed according to a Poisson with parameter $\theta$, and the expected waiting time between important agents distributes as an exponential with rate $\theta$. If the algorithm stops at an important agent at time $t$, it gets an expected profit of $\lambda t + 2-\exp(-\theta t/2)$.}
\label{fig:iid2BCA}
\end{figure}
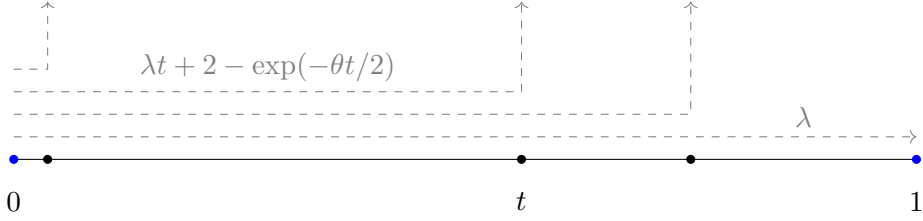

Let $\beta(t)$ be the weight that an algorithm that only has $t$ time units left and that has already selected one edge from the 4-cycle (the particle is at $1-t$ and moves towards 1 at unit speed). 
The only way for the algorithm to get more weight in $(1-t,1]$ is if another agent shows the edge opposite to the selected one. 
As important agents arrive at rate $\theta$ and each has a probability of 1/2 to show the opposite edge, the probability of seeing the opposite edge is the probability that an exponential with rate $\theta/2$ is less than $1-t$. Therefore, we have $\beta(t)= 1-e^{-\theta t/2}$.

Let $\alpha(t)$ be the weight an algorithm can get from the interval $(1-t,1]$ when it did not select any edge from the 4-cycle. 
Then $\alpha$ satisfies $\alpha(0)=0$ and for $t\in (0,1]$ we have 
\begin{align}
\alpha(t)&=\lambda te^{-\theta t}+ \int_{0}^{t}(\lambda x+\max(\alpha(t-x),1+\beta(t-x)) \theta e^{-\theta x}\mathrm dx \notag\\
&=\lambda te^{-\theta t}+\int_0^{t} \lambda x\theta e^{-\theta x}\mathrm dx + \int_{0}^{t}\max(\alpha(t-x),2-e^{-\frac{\theta}{2}(t-x)})\theta e^{-\theta x}\mathrm dx \notag\\
&=\frac{\lambda}{\theta} (1 - e^{-\theta t})+ \theta e^{-\theta t}\int_{0}^{t}\max(\alpha(s),2-e^{-\frac{\theta}{2}s})e^{\theta s}\mathrm ds, \label{iid-alpha}
\end{align} 
where the first equality holds by considering the scenarios in which the agent selects an edge in the cycle or not in the interval $(1-t,1]$, the second equality holds by replacing with the value of $\beta(t-x)$, and the third holds by evaluating the first integral and changing variables in the second.
The curves $\alpha(s)$ and $1+\beta(s)=2-e^{-\theta s/2}$ are increasing, and $0=\alpha(0)\leq 1+\beta(0)=1$.
Let $s^{\star}$ be the first point in $[0,1]$, if it exists, such that $\alpha(s^{\star})=1+\beta(s^{\star})$; if the curves do not intersect, we set $s^{\star}=1$.
Then, for $t\in [0,s^*]$, by evaluating in \eqref{iid-alpha} we have
\begin{align*}
    \alpha(t)&=\frac{\lambda}{\theta} (1 - e^{-\theta t})+ \theta e^{-\theta t}\int_0^t(2-e^{-\theta s/2})e^{\theta s}\mathrm ds =\frac{\lambda}{\theta} (1 - e^{-\theta t})+ 2(1-e^{-\theta t/2}),
\end{align*}
and therefore, if $s^*<1$, by using the above expression and solving the equation $\alpha(s^{\star})=1+\beta(s^{\star})$ we get 
$$s^{\star}=(2/\theta)\ln((\theta + \sqrt{4 \lambda^2 + \theta^2})/(2\lambda)).$$
To obtain this value, we observe that the equation $\alpha(s^{\star})=1+\beta(s^{\star})$ is quadratic in $e^{-\theta s^{\star}/2}$, and from here we can find the two candidate solutions, one of the negative.
On the other hand, 
$$\lim_{t\to s^{\star}}\beta'(t)=\lim_{t\to s^{\star}} \theta e^{-\theta t/2}/2=\lambda \theta/(\theta+\sqrt{4\lambda^2+\theta^2})<\lambda=\lim_{t\to s^{\star}}\alpha'(t),$$
therefore, there exists a $\delta>0$ such that $\alpha(t)\ge 1+\beta(t)$ for every $t\in (s^{\star},s^{\star}+\delta)$. 
Then, from \eqref{iid-alpha}, for every $t\in(s^{\star},s^{\star}+\delta)$ we have
\begin{align*}
    \alpha(t)e^{\theta t}&= \frac{\lambda}{\theta}(e^{\theta t}-1) +\theta\int_{0}^{s^{\star}}(2-e^{-\frac{\theta}{2}s})e^{\theta s}ds + \theta\int_{s^{\star}}^t \alpha(s)e^{\theta s}ds,
\end{align*}
and therefore, by taking the derivative in the above expression, we get that $(\alpha(t)e^{\theta t})'=\lambda e^{\theta t}+\lambda \alpha(t) e^{\theta t}$.
Since $(\alpha(t)e^{\theta t})'=\alpha'(t)e^{\theta t}+\theta \alpha(t)e^{\theta}$, we get $\alpha'(t)=\lambda$ for every $t\in(s^{\star},s^{\star}+\delta)$.
Since $\alpha$ grows linearly from $s^{\star}$, and $\beta(t)$ is concave, we conclude that $\alpha(t)$  never crosses $1+\beta(t)$ again, and that $\alpha(t)=\alpha(s^{\star})+\lambda(t-s^{\star})$ for every $t\in (s^{\star},1]$.
Then, as long as $\lambda$ and $\theta$ are such that $s^{\star}<1$, the algorithm gets an expected weight of $\alpha(1)=\alpha(s^{\star})+\lambda(1-s^{\star})=2-e^{-\theta s^{\star}/2}+\lambda(1-s^{\star})$, which together with \eqref{iid-lim-opt} implies that the competitiveness of the algorithm is upper-bounded by
\begin{align*}
\frac{2-\frac{2\lambda}{\theta+\sqrt{4\lambda^2+\theta^2}} + \lambda\left(1-\frac{2}{\theta}\ln\left(\frac{\theta + \sqrt{4 \lambda^2 + \theta^2}}{2 \lambda})\right)\right)}{ \lambda +2+\frac{e^{-\theta}}{36}\Big(72-144e^{\theta/6}-12\theta-\theta^2\Big)}.
\end{align*}
Numerically, this expression is minimized in $\lambda\approx 1.4737$ and $\theta\approx 2.8224$, attaining a value of $\approx 0.686641$. 
Here $s^{\star}\approx 0.603076\in (0,1)$ as required, and therefore the competitiveness of any online algorithm is upper-bounded by $\approx 0.686641$.
\end{proof}

\noindent{\bf Acknowledgements.} ANID-Chile partially supported this research through grants CMM FB210005, FONDECYT 1231669, FONDECYT 1241846, and ANILLO ACT210005.

\bibliographystyle{acm}
{\small\bibliography{references}}

\end{document}